\documentclass{article}
\usepackage{ijcai26}

\usepackage{times}
\usepackage{soul}
\usepackage{url}
\usepackage[hidelinks]{hyperref}
\usepackage[utf8]{inputenc}
\usepackage[small]{caption}
\usepackage{graphicx}
\usepackage{amsmath}
\usepackage{amssymb}
\usepackage{amsthm}
\newtheorem{prop}{Proposition}
\usepackage{booktabs}
\usepackage{algorithm}
\usepackage{algorithmic}
\usepackage[switch]{lineno}

\newtheorem{theorem}{Theorem}

\usepackage{thm-restate}
\theoremstyle{definition}
\newtheorem{definition}[theorem]{Definition}

\usepackage[capitalise]{cleveref}
\crefname{prop}{Proposition}{Propositions}

\usepackage{bbm}
\usepackage{xcolor}

\title{Policy-Embedded Graph Expansion: Networked HIV Testing with Diffusion-Driven Network Samples}

\author{
Akseli Kangaslahti$^1$
\and
Davin Choo$^1$\and
Lingkai Kong$^1$\and
Milind Tambe$^1$\and\\
Alastair van Heerden$^{2,3}$\And
Cheryl Johnson$^4$
\\
\affiliations
$^1$Harvard University\\
$^2$University of Witwatersrand\\
$^3$Wits Health Consortium\\
$^4$World Health Organization
\emails
akselikangaslahti@g.harvard.edu,
davinchoo@seas.harvard.edu,
lingkaikong@g.harvard.edu,
tambe@g.harvard.edu,
alastair.vanheerden@wits.ac.za,
johnsonc@who.int
}

\begin{document}

\maketitle

\begin{abstract}

HIV is a retrovirus that attacks the human immune system and can lead to death without proper treatment.
In collaboration with the WHO and the University of Witwatersrand, we study how to improve the efficiency of HIV testing with the goal of eventual deployment, directly supporting progress toward UN Sustainable Development Goal 3.3.
While prior work has demonstrated the promise of intelligent algorithms for sequential, network-based HIV testing, existing approaches rely on assumptions that are impractical in our real-world implementations.
Here, we study sequential testing on incrementally revealed disease networks and introduce Policy-Embedded Graph Expansion (PEGE), a novel framework that directly embeds a generative distribution over graph expansions into the decision-making policy rather than attempting explicit topological reconstruction.
We further propose Dynamics-Driven Branching (DDB), a diffusion-based graph expansion model that supports decision making in PEGE and is designed for data-limited settings where forest structures arise naturally, as in our real-world referral process.
Experiments on real HIV transmission networks show that the combined approach (PEGE + DDB) consistently outperforms baselines (e.g., 17.3\% improvement in discounted reward and 15.4\% more HIV detections with 25\% of the population tested) and explore key tradeoffs that drive solution quality.

\end{abstract}

\section{Introduction} \label {introduction}

Human immunodeficiency virus (HIV) attacks the human immune system and, without treatment, can lead to acquired immunodeficiency syndrome (AIDS). Although there is currently no cure available for HIV, proper medical treatment can help control the virus and its spread. UNAIDS has proposed the 95-95-95 targets for HIV that aim for 95\% of individuals with HIV to be aware of their status, 95\% of those aware of their positive status to receive treatment, and 95\% of people that are treated to reach viral suppression \cite{unaids2022}. However, a more recent UNAIDS report has shown that the first of these goals has proven to be the most difficult to reach, as about 1 in 7 individuals with HIV have not been diagnosed and 1.3 million new infections occur each year \cite{unaids2024}. Furthermore, with recent funding shifts, projections expect these gaps to grow further \cite{ten2025impact}. Therefore, in line with the UN Sustainable Development Goal 3.3 \cite{UN_SDG3}, we aim to help people with HIV get tested as quickly as possible so that they can begin receiving treatment and eventually achieve viral suppression. This also facilitates efficient delivery of HIV prevention such as pre-exposure prophylaxis (PrEP) to high-risk individuals with connections to people who test positive. 

In some countries like South Africa with high HIV prevalence, the concept of Universal Test and Treat (UTT) is in place. In UTT, everyone in a community is offered HIV testing and anyone who tests positive is immediately started on antiretroviral therapy (ART) to suppress the virus, reduce transmission, and improve individual health. However, testing resources such as test kits and staff can be limited in practice, so we need to develop approaches for determining an efficient testing order. The World Health Organization (WHO) has recommended network-based testing (NBT) for HIV \cite{who}, which has proven to be effective in South Africa \cite{jubilee2019hiv} and for infectious diseases beyond HIV (e.g., \cite{juher2017network,monroe2025can}). The goal of NBT is to identify individuals at risk for STIs by engaging their social, sexual, or family networks. This strategy helps reach stigmatized populations that are often inaccessible via centralized outreach through peer referrals that leverage existing social trust between individuals. A prominent instantiation is the Enhanced Peer Outreach Approach (EPOA), where influential ``seeds'' from key populations use performance-based incentives to recruit high-risk contacts, effectively bridging the gap between hidden networks and clinical services \cite{fhi}. 
However, the complexity of NBT makes it difficult for humans to determine optimal decisions by hand, motivating the need for an intelligent decision-making framework.
Prior work has proposed network-based testing strategies for HIV \cite{zhang2023adaptive,choo2025adaptive} and other diseases like COVID-19 \cite{cui2021network}. However, these approaches require full network observability, which is prohibitively expensive to collect in practice. While some work on unknown networks exists, many efforts target different testing problems \cite{mcfall2021optimizing} or base recommendations on time of infection data, which we do not have access to \cite{nikolopoulos2016network,morgan2019network}. 

 With these considerations in mind, we study the problem of allocating a limited number of HIV tests across individuals in a network to identify as many HIV-positive individuals as possible, as quickly as possible. Specifically, we consider a sequential frontier-based formulation (see Section \ref{problem_statement}) where the network is revealed incrementally as test recipients refer contacts to testing services. Furthermore, as in many health domains, we need to adapt to small and noisy training data, which is an especially difficult challenge when combined with the high-dimensional node covariates that we observe and the combinatorial explosions that are natural in graph-based problems.

\subsection{Contributions} \label{intro:contributions}
In this paper, we study how generative graph expansion can be embedded in the sequential decision-making loop to tackle the limited observability of the disease network and data limitations. Specifically, our contributions are:\\
\textbf{1)} We formalize our HIV testing setting into a new sequential decision-making problem on incrementally revealed graphs, which has applications beyond HIV testing.
\\
\textbf{2)} We introduce a new approach called Policy-Embedded Graph Expansion (PEGE) for solving sequential decision-making problems on incrementally revealed graphs. PEGE is the first approach to this class of problems that directly embeds a generative graph expansion model into the sequential decision-making policy loop, enabling intelligent decisions that respect network uncertainty without requiring full topological accuracy on hidden network reconstruction. PEGE is well suited for problems with data limitations and/or high variance in the realization of the unobservable graph.
\\
\textbf{3)} We further propose Dynamics Driven Branching (DDB), a new generative graph expansion model tailored to data-scarce settings where forest topologies occur naturally, e.g., our network referral system, that trains without graph-level supervision and supports decision-making in PEGE. DDB is a hybrid diffusion and Gaussian Process Regression (GPR) model that learns to capture the node-to-node dynamics that drive variation in the unobservable portion of the network.
\\
\textbf{4)} We empirically evaluate our approach (PEGE + DDB), along with several baseline algorithms, showing that it outperforms all baselines on real-world HIV networks. We also provide additional experiments that highlight key tradeoffs that drive solution quality.

\subsection{Interdisciplinary Collaboration and Impact}  \label{intro:collaboration}
Our interdisciplinary research team consists of both AI researchers and domain experts from the WHO and the University of Witwatersrand. This work on network-based testing strategies for intelligent HIV diagnosis and prevention is a part of a larger ongoing collaboration to study how AI can be used to support HIV prevention and treatment efforts in South Africa. The domain experts on our team have played a crucial and iterative role in designing the problem formulation to make it as practical as possible, assessing real-world applicability of proposed approaches, and connecting our research to avenues for deployment.  

\begin{figure}
    \centering
    \includegraphics[width=\linewidth]{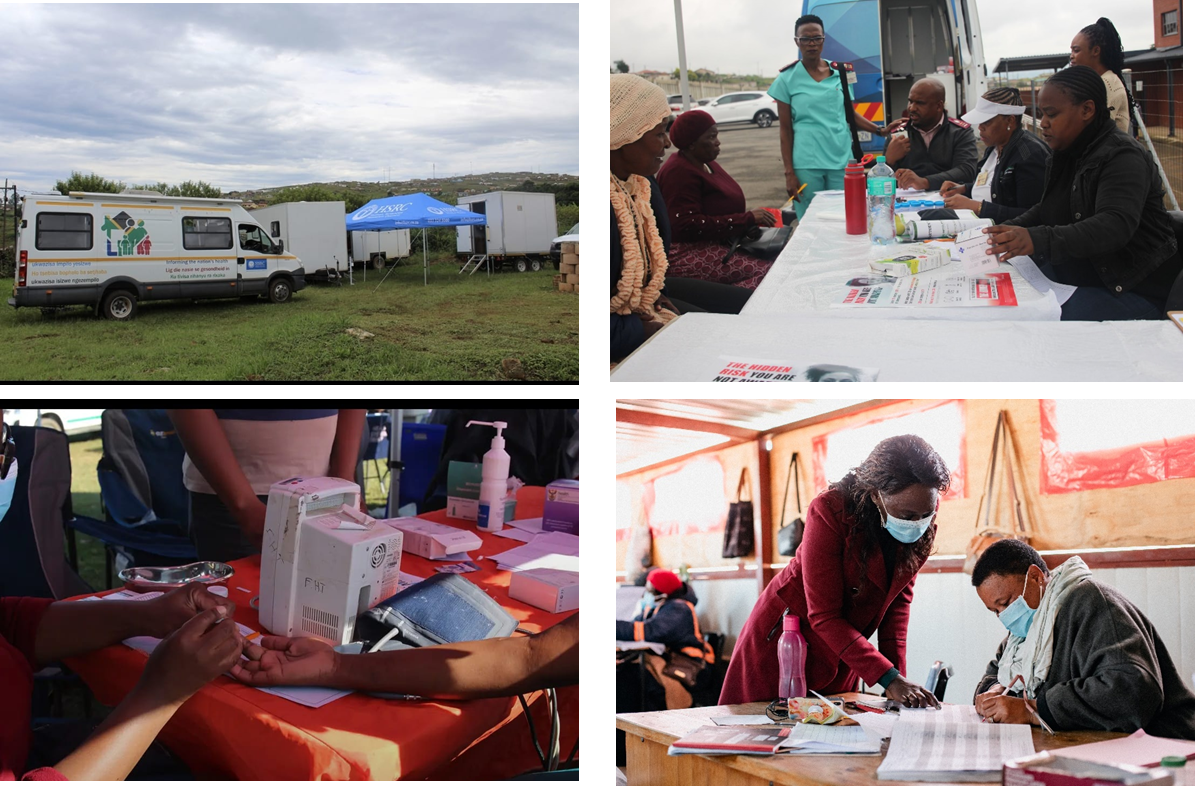}
    \caption{HIV self-test kit distribution in South Africa.}
    \label{fig:test_distribution}
\end{figure}

We are currently aiming to deploy the work presented in this paper to have a real-world impact in improving our HIV testing efforts in South Africa. Figure \ref{fig:test_distribution} shows the team distributing HIV self-test kits in the region. Specifically, we want to integrate our methods to determine testing sequences in clinics with limited testing resources like test kits and staff.

\section{Related Work} \label{related_work}
\subsection{Network-Based Disease Testing} \label{related:application}
Network-based modeling of HIV and other infectious diseases is a well-studied field \cite{mattie2022review}. One area is in algorithms for selecting groups from known networks for pooled testing to maximize resource efficiency \cite{silva2021group,sewell2022leveraging}, with some efforts incorporating dynamically updating networks \cite{srinivasavaradhan2024dynamic}. Another effort has proposed a Gittins index-based policy for sequential frontier-based testing that we leverage in this work \cite{choo2025adaptive}. However, these efforts assume full network observability, which is often unrealistic in practice. Other studies on unknown networks have either focused on other types of testing-based problem formulations \cite{mcfall2021optimizing,tan2025deeptrace} or on higher-level strategies reliant on temporal infection data \cite{nikolopoulos2016network,morgan2019network}, which are not applicable to our setting.

\subsection{Sequential Decision-Making on Incrementally Revealed Graphs} \label{related:algs}
Many types of sequential decision-making problems on incrementally revealed graphs have close ties to the problem that we study in this work (see Section \ref{problem_statement}). For example, in active search on graphs, the goal is to query as many nodes with a hidden target label as possible given a budget and graph topology \cite{wang2013active,tsui2024optimal}. Active Learning is a common approach to such problems. Other areas include online graph exploration, where the goal is to discover all nodes in an initially unknown graph with minimal path cost \cite{megow2012online}, for which learning-augmented and reinforcement learning approaches have been proposed \cite{eberle2022robustification,chiotellis2020neural}. Another relevant problem class is influence maximization (IM) with partial or myopic feedback, where the goal is to sequentially select a small set of root/seed nodes in a network that maximizes influence (e.g., spread of information) in a network \cite{yuan2016no,peng2019adaptive}. In fact, IM has been used in real network-based HIV interventions \cite{wilder2021clinical}, although these algorithms do not apply to our settings since we cannot query the network to learn more information. One closely related effort in IM shows how partial graph observability can be aided via maximum likelihood link prediction conditioned on known node metadata \cite{tran2024meta}, solving an intermediate learning task to address partial observability similarly to PEGE. Similar approaches have been proposed in online graph exploration \cite{tan2025deeptrace}. Other efforts in robotic navigation have employed intermediate learning to supplement information about uncertain frontier nodes in a graph \cite{cui2024frontier}. Decision-focused learning is an alternative strategy where intermediate learning tasks such as link prediction can be directly optimized with respect to downstream decision-making outcomes \cite{wilder2019end}. In contrast, in PEGE, the generated graph expansions represent a distribution rather than a single prediction. The model is not trained on intermediate topological accuracy or downstream decision quality; rather, its goal is to iteratively capture variation in decision-relevant areas of a graph, which change as the graph grows incrementally. This makes PEGE more suitable for high-dimensional and data-scarce settings with high variance in the realization of the unobservable graph, where topology prediction is challenging and learned decision-focused dynamics from training may not transfer well to test graphs.

\section{Problem Definition} \label{problem_statement}
We first define a more abstract problem class and solution approach for generalization beyond HIV testing, and in Section \ref{application} show how it applies to our HIV testing application. Specifically,
we define \emph{Sequential Acting on Partially Observed Graphs (SAPOG)} as a class of sequential decision-making problems in which an agent acts on vertices of a fixed but initially partially observed graph, with actions revealing additional structure and node information over time.

Throughout this section, we will use $\subseteq$ to denote subsets and subgraphs, $|\cdot|$ to denote set cardinality, $N(v)$ to refer to the neighbors of a node $v$, and $\setminus$ to denote set subtraction. We will also refer to node attributes as covariates.

\begin{definition}[Sequential Acting on Partially Observed Graphs (SAPOG)]
A SAPOG instance is specified by the tuple $\mathcal{P} = (\mathcal{G}, \mathbf{X}, \mathbf{Y}, s^{(0)}, R, \gamma, T)$, where the fixed ground truth graph $\mathcal{G} = (\mathbf{V}, \mathbf{E})$ is revealed incrementally over $T \in \mathbb{N}_{\geq 1}$ rounds with respect to discount factor $\gamma \in (0,1]$.
Functions $\mathbf{X}: \mathbf{V} \to \mathbb{R}^{d_{c}}$ and $\mathbf{Y}: \mathbf{V} \to \mathbb{R}^{d_{\ell}}$ map each node to covariates and labels (e.g., binary HIV status) respectively.
At any timestep $t \in [T]$, the state $s^{(t)} = (\mathbf{V}^{(t)}, \mathbf{V}^{(t)}_\mathbf{Y})$ is a pair of vertex subsets with $\mathbf{V}^{(t)}_\mathbf{Y} \subseteq \mathbf{V}^{(t)} \subseteq \mathbf{V}$, where $\mathbf{V}^{(t)}$ is the set of nodes discovered up to time $t$ and $\mathbf{V}^{(t)}_\mathbf{Y}$ is the subset with revealed labels.
In other words, given state $s^{(t)}$, we observe the induced subgraph $\mathcal{G}^{(t)} = \mathcal{G}[\mathbf{V}^{(t)}]$, along with covariates $\mathbf{X}^{(t)} = \{\mathbf{X}(v) \in \mathbb{R}^{d_{c}} : v \in \mathbf{V}^{(t)}\}$ and labels $\mathbf{Y}^{(t)} = \{\mathbf{Y}(v) \in \mathbb{R}^{d_{\ell}} : v \in \mathbf{V}^{(t)}_\mathbf{Y}\}$.
Upon picking an action from the frontier $a \in \mathbf{V}^{(t)} \setminus \mathbf{V}^{(t)}_\mathbf{Y}$, we reveal the label of node $a$ and its previously undiscovered neighbors, i.e., $\mathbf{V}^{(t+1)} = \mathbf{V}^{(t)} \cup N(a)$ and $\mathbf{V}^{(t+1)}_\mathbf{Y} = \mathbf{V}^{(t)}_\mathbf{Y} \cup \{a\}$.
Given some initial state $s^{(0)} = (\mathbf{V}^{(0)}, \emptyset)$, the goal is to compute a policy $\pi$ that maps each state $s^{(t)}$ to a node in $\mathbf{V}^{(t)}$, maximizing the expected total discounted reward:
\begin{equation}
\label{eq:reward-model}
\pi^* = \arg\max_\pi \mathbb{E}_{\pi} \left[ \sum_{t=0}^{T-1} \gamma^{t} \cdot R \big( s^{(t)}, \pi(s^{(t)}) \big) \right]
\end{equation}
where $\pi(s^{(t)}) \in \mathbf{V}^{(t)}\setminus \mathbf{V}^{(t)}_\mathbf{Y}$ is the frontier node selected by policy $\pi$ and $R(s^{(t)}, v)$ is the (possibly label-dependent) reward for acting on node $v = \pi(s^{(t)})$ at round $t$.
\end{definition}

\begin{figure}
    \centering
    \includegraphics[width=\linewidth]{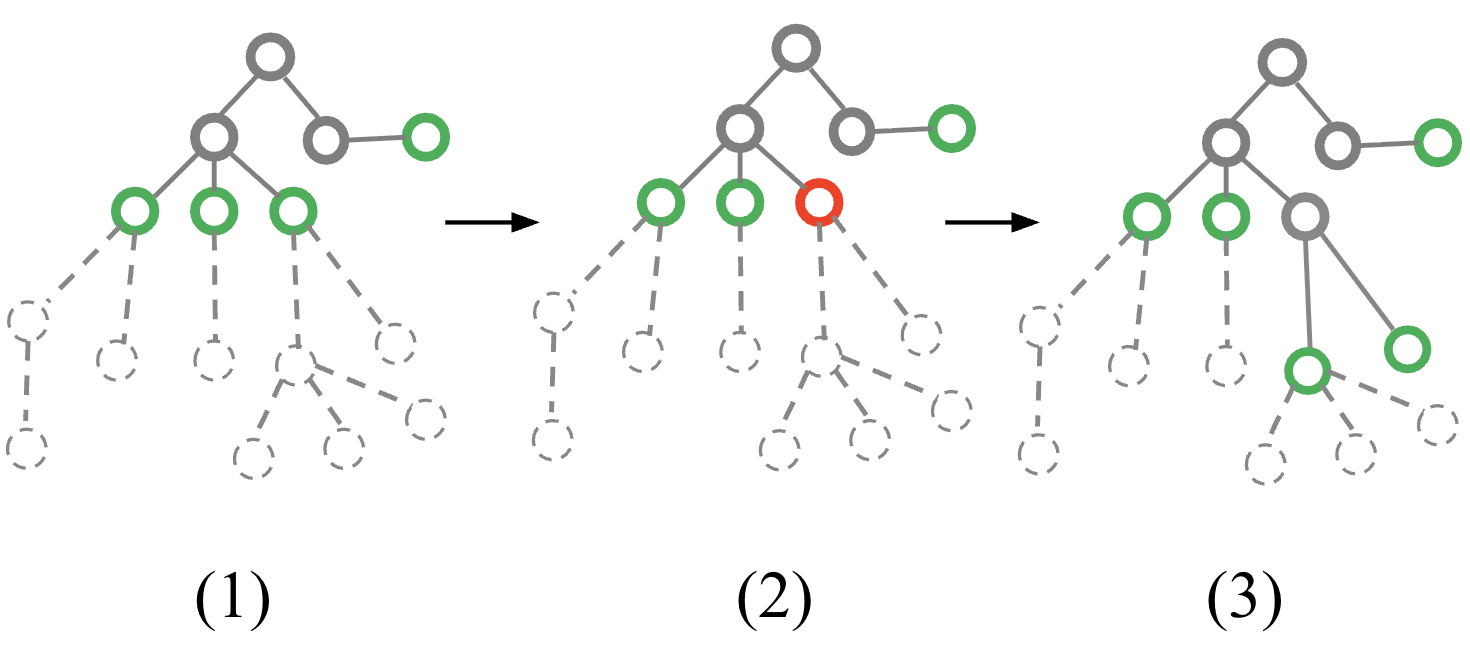}
    \caption{An illustration of a SAPOG instance at time $t$.
    Solid nodes and edges represent the induced subgraph $\mathcal{G}^{(t)} = \mathcal{G}[\mathbf{V}^{(t)}]$.
    Grey nodes represent nodes that have already been acted on, i.e., $\mathbf{V}^{(t)}_\mathbf{Y}$. Green nodes represent frontier nodes, which are nodes that have been discovered but not acted on, and thus have an unrevealed status. These frontier nodes make up the action space $\mathbf{V}^{(t)} \setminus\mathbf{V}^{(t)}_\mathbf{Y}$.
    Dashed nodes and edges represent unrevealed nodes, i.e., $\mathbf{V} \setminus \mathbf{V}^{(t)}$ and their unrevealed incident edges.
    (1) Agent evaluates the state $s^{(t)}$.
    (2) Agent acts on the red node $a \in \mathbf{V}^{(t)} \setminus\mathbf{V}_\mathbf{Y}^{(t)}$.
    (3) Agent observes the label of node $a$ and reveals its previously undiscovered neighbors in the next state.} 
    \label{fig:sapog}
\end{figure}

Figure \ref{fig:sapog} shows an example timestep in a SAPOG instance. The SAPOG class captures several other existing problem formulations. For example, Adaptive Frontier Exploration on Graphs \cite{choo2025adaptive} is a special case where $\mathbf{V}^{(0)} = \mathbf{V}$. Additionally, active search on graphs \cite{wang2013active} and selective harvesting \cite{murai2018selective} are example undiscounted cases where $\gamma = 1$ and $T < |\mathbf{V}|$ and, in active search, $\mathbf{V}^{(0)} = \mathbf{V}$. The frontier-based decision-making in SAPOG is also analogous to other fields where local actions reveal information about the neighboring environment, such as robotic exploration \cite{keidar2014efficient}.

\section{Methods} \label{methods}
\subsection{SAPOG Approach Analysis} \label{methods:analysis}
There are several challenges in solving SAPOG instances. First, since $(\mathcal{G}, \mathbf{X}, \mathbf{Y}) \setminus (\mathcal{G}^{(t)}, \mathbf{X}^{(t)}, \mathbf{Y}^{(t)})$ is unknown, the agent needs to make decisions without knowledge of their long-term impact or even their short-term rewards. Second, it is common that in practice, there are multiple different possible realizations of the completed graph, i.e., the agent should view $(\mathcal{G}, \mathbf{X}, \mathbf{Y})$ as a distribution conditioned on $(\mathcal{G}^{(t)}, \mathbf{X}^{(t)}, \mathbf{Y}^{(t)})$. While algorithms for the case where $\mathbf{V}^{(0)} = \mathbf{V}$ can be formulated, they are difficult to apply to general SAPOG instances where the graph and covariates are only partially observable and incrementally revealed.

Prior work has shown that even in the fully observable case where $\mathbf{V}^{(0)} = \mathbf{V}$, directly learning the Q function of nodes is not an effective approach \cite{choo2025adaptive}, motivating some intermediate learning step to address the uncertainty in the unobservable portion of the graph. One natural approach is used in prior work for similar decision-making problems on partially observable graphs \cite{tran2024meta}. In the context of SAPOG, the analogous approach would be to learn to predict the maximum-likelihood graph and covariate completion conditioned on the observable $(\mathcal{G}^{(t)}, \mathbf{X}^{(t)}, \mathbf{Y}^{(t)})$, then iteratively plug the learned completion into an oracle that solves for optimal actions under the assumption that the learned completion is correct. More formally, we define such a decision-making oracle $\mathcal{O}$ as an algorithm that returns a real-valued evaluation (e.g., Q-value) of any action $v$ given $\mathbf{Y}^{(t)}$ and any $(\mathcal{G}^\prime \supseteq \mathcal{G}^{(t)}, \mathbf{X}^\prime \supseteq \mathbf{X}^{(t)})$, i.e., any possible completion of the observed $(\mathcal{G}^{(t)}, \mathbf{X}^{(t)})$ that is assumed by $\mathcal{O}$ to be the ground truth. However, we show in Proposition \ref{prop:non-optimality} that this strategy is not necessarily effective in SAPOG.

\begin{prop}[Non-Optimality of Maximum-Likelihood Graph Completion]
\label{prop:non-optimality}
Assume that, during any timestep $t$ of a SAPOG instance, an agent:
\begin{enumerate}
    \item Has access to a decision-making oracle $\mathcal{O}$ that returns the Q-value of given actions.
    \item Has access to $(\mathcal{G}^\star \supseteq \mathcal{G}^{(t)}, \mathbf{X}^\star \supseteq \mathbf{X}^{(t)})$, the maximum-likelihood completion of the observed $(\mathcal{G}^{(t)}, \mathbf{X}^{(t)})$.
\end{enumerate}
Then, the action $v^\star = \underset{v \in \mathbf{V}^{(t)}\setminus\mathbf{V}^{(t)}_\mathbf{Y}}{\arg\max } \mathcal{O}(v \mid \mathbf{Y}^{(t)}, \mathcal{G}^\star, \mathbf{X}^\star)$, is not necessarily optimal in expectation. 
\end{prop}
\begin{proof}
Suppose that $(\mathcal{G}_1, \mathbf{X}_1), (\mathcal{G}_2, \mathbf{X}_2)$ are the only two possible completions of $(\mathcal{G}^{(t)}, \mathbf{X}^{(t)})$, with $\Pr(\mathcal{G}_1, \mathbf{X}_1 \mid \mathbf{Y}^{(t)}, \mathcal{G}^{(t)}, \mathbf{X}^{(t)}) = \frac{3}{5}$ and $\Pr(\mathcal{G}_2, \mathbf{X}_2 \mid \mathbf{Y}^{(t)}, \mathcal{G}^{(t)}, \mathbf{X}^{(t)}) = \frac{2}{5}$.
Thus, $\mathcal{G}^\star = \mathcal{G}_1$ and $\mathbf{X}^\star = \mathbf{X}_1$.
Now, for some constant $c > 0$, suppose that there exist two possible actions $v_1, v_2 \in \mathbf{V}^{(t)} \setminus\mathbf{V}^{(t)}_\mathbf{Y}$ such that $\mathcal{O}(v_1 \mid \mathbf{Y}^{(t)},\mathcal{G}_1, \mathbf{X}_1) = 2c$, $\mathcal{O}(v_2 \mid \mathbf{Y}^{(t)},\mathcal{G}_1, \mathbf{X}_1) = c$, $\mathcal{O}(v_1 \mid \mathbf{Y}^{(t)},\mathcal{G}_2, \mathbf{X}_2) = c$, and $\mathcal{O}(v_2 \mid \mathbf{Y}^{(t)},\mathcal{G}_2, \mathbf{X}_2) = 10c$.
Then, $v^\star = v_1$, which is suboptimal since the expected Q-value $v_1$ is $\frac{8c}{5}$ while $v_2$ has a larger expected Q-value of $\frac{23c}{5}$.
\end{proof}

Let $\mathcal{D}^{(t)}$ be the true distribution of possible graph completions given $(\mathbf{Y}^{(t)}, \mathcal{G}, \mathbf{X}^{(t)})$.
Since we do not know the true future, we wish to evaluate an action $v$ by maximizing
\begin{equation}
\label{eq:expected-objective}
\mathbb{E}_{(\mathcal{G}^\prime, \mathbf{X}^\prime) \sim \mathcal{D}^{(t)}}[\mathcal{O}(v \mid \mathbf{Y}^{(t)}, \mathcal{G}^\prime, \mathbf{X}^\prime)]
\end{equation}
As shown in \cref{prop:non-optimality}, we should not decide on the action via $\mathcal{O}(v \mid \mathbf{Y}^{(t)}, \text{mode}(\mathcal{D}^{(t)}))$, where $\text{mode}(\mathcal{D}^{(t)})$ denotes the maximum-likelihood completion in $\mathcal{D}^{(t)}$, but instead learn the distribution $\mathcal{D}^{(t)}$ itself.
This would then allow us to approximate \cref{eq:expected-objective} via samples. For this task, we consider a generative graph expansion model $\mathcal{M}$ that, given any state $s^{(t)} = (\mathcal{G}^{(t)}, \mathbf{X}^{(t)}, \mathbf{Y}^{(t)})$, can be sampled to produce a graph expansion $\bar{s}^{(t)} = (\mathcal{G}^\prime, \mathbf{X}^\prime)$ such that $\mathcal{G}^{(t)} \subseteq \mathcal{G}^\prime$ and $\mathbf{X}^{(t)} \subseteq \mathbf{X}^\prime$. We will see in Section \ref{methods:pege} that $\mathcal{M}$ does not need to accurately predict the fully topologies in $\mathcal{D}^{(t)}$ in order to support decision-making in PEGE.

\subsection{Policy-Embedded Graph Expansion (PEGE)} \label{methods:pege}

With the challenges of SAPOG and our analysis in mind, we propose Policy-Embedded Graph Expansion (PEGE), a new approach to sequential decision-making problems over incrementally revealed graphs. To the best of our knowledge, PEGE is the first attempt at leveraging a generative graph expansion model to address limited graph observability in the sequential decision-making loop for an incrementally revealed graph. PEGE solves SAPOG instances as described in Algorithm \ref{alg:pege} and illustrated in Figure \ref{fig:PEGE}.

\begin{algorithm}[t!]
\caption{Policy-Embedded Graph Expansion (PEGE)}
\label{alg:pege}
\begin{algorithmic}[1]
\REQUIRE SAPOG instance $\mathcal{P} = (\mathcal{G}, \mathbf{X}, \mathbf{Y}, s^{(0)}, R, \gamma, T)$, sampling parameter $k \in \mathbb{N}$, generative graph expansion model $\mathcal{M}$, decision-making oracle $\mathcal{O}$
\ENSURE Actions $a^{(1)}, \ldots, a^{(T)}$ for each time step $t \in [T]$
\FOR{$t = 0, 1, \ldots, T-1$}
    \STATE Observe state $s^{(t)} = (\mathcal{G}^{(t)}, \mathbf{X}^{(t)}, \mathbf{Y}^{(t)})$
    \STATE Using $\mathcal{M}(s^{(t)})$, sample $k$ graph completions $\bar{s}^{(t)}_1, \ldots, \bar{s}^{(t)}_k$ where $\bar{s}^{(t)}_j = (\mathcal{G}_j^{(t)}, \mathbf{X}_j^{(t)})$
    \STATE For all actions $a \in \mathbf{V}^{(t)} \setminus \mathbf{V}^{(t)}_\mathbf{Y}$, define $f(a)$ as the empirical estimate of $\mathbb{E}_{\bar{s}^{(j)} \sim \mathcal{M}(s^{(t)})}[\mathcal{O}(a \mid \mathbf{Y}^{(t)}, \bar{s}^{(j)})]$:
    \begin{align*}
    f(a)
    = \frac{1}{k} \sum_{j=1}^k \mathcal{O}(a \mid \mathbf{Y}^{(t)}, \bar{s}^{(t)}_j)
    \end{align*}
    \STATE Take action $a^{(t+1)} \in \underset{a \in \mathbf{V}^{(t)}\setminus\mathbf{V}^{(t)}_\mathbf{Y}}{\arg\max } f(a)$
\ENDFOR
\STATE \textbf{return} $a^{(1)}, \ldots, a^{(T)}$
\end{algorithmic}
\end{algorithm}

\begin{figure}
    \centering
    \includegraphics[width=\linewidth]{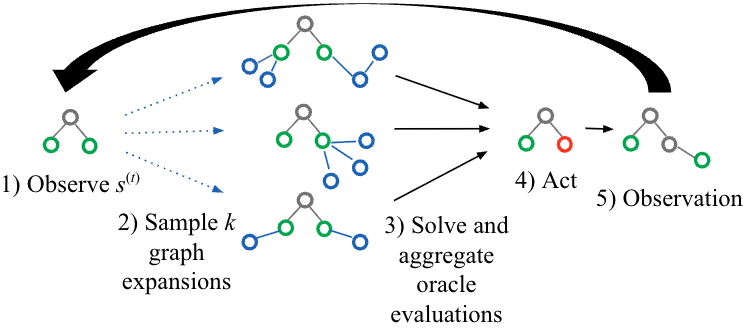}
    \caption{Workflow Diagram for PEGE. Grey nodes represent $\mathbf{V}^{(t)}_{\mathbf{Y}}$ while green nodes represent frontier nodes, i.e., $\mathbf{V}^{(t)} \setminus\mathbf{V}^{(t)}_\mathbf{Y}$. The red node represents an action. Blue nodes and edges represent graph expansions generated by $\mathcal{M}$ that are treated as true observed nodes and covariates by the oracle. 1) The state $s^{(t)}$ is observed.  2) We sample $k$ graph expansions up to depth $d$ from frontier nodes. In our case, we use a diffusion-driven model for $\mathcal{M}$, as described in Section \ref{application}. In this example, $k = 3$ and $d = 1$. 3) The oracle solves each of the $k$ instances separately and the resulting action evaluations are aggregated. 4) The agent takes an action $v$. 5) Reward, the label of $v$, and new nodes adjacent to $v$ and their covariates are observed as we transition to the next state.}
    \label{fig:PEGE}
\end{figure}

The generative graph expansion model $\mathcal{M}$ and the decision-making oracle $\mathcal{O}$ (defined in Section \ref{methods:analysis}) have a closely intertwined relationship: the graph expansions are plugged into $\mathcal{O}$ to decide which node to test on timestep $t$, which influences which nodes $\mathcal{M}$ expands from during timestep $t + 1$ after neighboring true graph nodes are revealed. Further, there can be less of a burden on $\mathcal{M}$ of learning the full distribution $\mathcal{D}^{(t)}$ conditioned on the partially observed graph since its output is only used for downstream discounted decision-making. Normally, learning $\mathcal{D}^{(t)}$ is very challenging, as $|\mathbf{V}|$ is unknown, so the model has no sense of how large $\mathcal{G}$ is. However, with a discount factor $\gamma < 1$ in place, generated nodes become less relevant to decision-making as they get further from the frontier nodes. Thus, we can limit the graph expansion to a maximum depth $d$ from frontier nodes to help address the unknown size of the complete $\mathcal{G}$ while ensuring that $\mathcal{M}$ can still capture the variation in possible graph completions that most strongly impacts decision-making. $\mathcal{M}$ must be able to adapt when these decision-relevant areas of the graph (i.e., the areas near the frontier) change as the graph is incrementally revealed and the frontier updates. The Monte-Carlo nature of PEGE helps produce actions that are more robust to variation in this decision-relevant portion of $(\mathcal{G}, \mathbf{X}, \mathbf{Y})$ conditioned on $(\mathcal{G}^{(t)}, \mathbf{X}^{(t)}, \mathbf{Y}^{(t)})$.

\section{Application to Network-Based HIV Testing} \label{application}
In our sequential network-based HIV testing setting (see Section \ref{intro:collaboration}), we are interested in producing an efficient testing order that identifies as many HIV-positive individuals as possible as quickly as possible. With budgeted testing kits and staff, this makes efficient use of limited testing resources, helps people with HIV and associated high-risk contacts get started with treatment and prevention as early as possible, and limits further spread of the virus. In this application, $\mathcal{G}$ represents a sexual contact network that grows incrementally from a set of root individuals $\mathbf{V}^{(0)}$ with initially unknown statuses ($\mathbf{V}^{(0)}_\mathbf{Y} = \emptyset$) as we allocate tests and test recipients refer contacts. One additional and important property in our application is that $\mathcal{G}$ is known to be a forest (see Section \ref{evaluation:dataset}). $\mathbf{X}$ represents binary covariates that are collected about each individual in the network, e.g., gender and occupation, ($d_c = 72$) and $\mathbf{Y}$ represents binary HIV status ($d_\ell$ = 1). Our reward function is simply $R(a) = \mathbf{Y}(a)$. For the oracle $\mathcal{O}$ used in PEGE, we use a Gittins index-based policy that produces real-valued evaluations of nodes and has proven performance guarantees when $\mathcal{G}$ is a forest \cite{choo2025adaptive}. This oracle learns a joint distribution of infection statuses, modeled via potential functions of individual covariates; see Appendix B of \cite{choo2025adaptive} for details. At test time, the oracle uses the learned function parameters $\theta$ to compute scores via a discrete version of the branching bandit formulation of \cite{keller2003branching}, and selects the node with the largest score.

However, one significant challenge remains with applying PEGE to our application: the choice of the generative graph expansion model $\mathcal{M}$. While off-the-shelf models may be more applicable in different problem instances, we have quite specific requirements for $\mathcal{M}$ given the PEGE setting, our application of the SAPOG problem, and our small, noisy dataset (see Section \ref{evaluation:dataset}):\\
\textbf{1)} $\mathcal{M}$ must be a generative model that can produce a distribution of possible graph expansions.\\
\textbf{2)} $\mathcal{M}$ must be able to generate both edges and high-dimensional node covariates.\\
\textbf{3)} $\mathcal{M}$ must be adaptable to various sizes of known partial graphs and tunable expansion depth limits ($d$).\\
\textbf{4)} $\mathcal{M}$ must preserve the forest property in the graphs that it generates (see Section \ref{evaluation:dataset}).\\
\textbf{5)} $\mathcal{M}$ must be suited to data-scarce settings. \\
No current model that we are aware of meets all of these criteria. Thus, we design Dynamics-Driven Branching (DDB), a new generative graph expansion model that trains only on low-level node-to-node dynamics with no graph or sequence-level supervision. DDB is a hybrid model that fuses denoising diffusion with Gaussian Process Regression (GPR) and learns to recursively expand branches off of the observed frontier nodes. The forest prior of $\mathcal{G}$ and the recursive property of DDB allow us to tabularize data from training graphs, reducing training of DDB to only node-to-node dynamics. Specifically, let $v_{\text{parent}}$ and $v_{\text{child}}$ be a pair of parent and child nodes with covariates $\mathbf{x}_{\text{parent}}$ and $\mathbf{x}_{\text{child}}$, respectively. Then:
\hspace{0pt}\\
\textbf{1)} We train a GPR model \cite{scikit-learn} to produce a Gaussian distribution of $|N(v_{\text{parent}})|$ conditioned on $\mathbf{x}_{\text{parent}}$.\\
\textbf{2)} We train a denoising diffusion model to learn a distribution of $\mathbf{x}_{\text{child}}$ conditioned on $\mathbf{x}_{\text{parent}}$. Our diffusion model employs a sinusoidal time embedding followed by four linear layers with GELU activations, trained via MSE loss.\\
At inference, DDB starts at each frontier node and recursively samples the GPR model to generate the number of successors at each expanded node, then samples the diffusion model to generate a vector of covariates for each of those successors. This way, DDB generates the stochastic node-to-node dynamics that drive variation in possible graph completions near frontier nodes, rather than directly learning specific structures and sequences present in the limited training data. 

The task of generating neighboring node covariates is especially challenging since the covariates have 72 dimensions and it is unclear which specific features in the parent node covariates are the most relevant. The diffusion model is thus a critical element of DDB, as its mode covering property helps learn an expressive distribution of the high-dimensional $\mathbf{x}_{\text{child}}$ conditioned on $\mathbf{x}_{\text{parent}}$. Additionally, targeting this task rather than learning to generate the entire graph in one shot helps limit training time of the diffusion model to only a few minutes and reduce overfitting or mode collapse.

Since DDB only expands new child nodes from parent nodes, it always preserves the forest structure of the input graph. Expansion ends when either the branching naturally terminates or the maximum node depth of $d$ from frontier nodes is reached. DDB is sampled to produce the $k$ graph expansions used in PEGE. Since DDB is conditioned on frontier nodes only and has a tunable expansion depth $d$, it is indifferent to the size of the known partial graph and can generate graph expansions of any specified depth.

The Gittins index-based oracle operates in $O(m^2)$ time on trees with $m$ nodes \cite{choo2025adaptive}.
In the worst-case, a frontier of $n$ nodes each generating $b^d$ descendants results in $m = nb^d$, yielding a per-sample complexity of $O(n^2 b^{2d})$.
However this is a very loose upper bound.
In practice, the frontier is a small subset of the $n$ nodes in the known graph, and the average branching factor is much lower than the maximum branching factor $b$.
Furthermore, our method is embarrassingly parallel: disjoint connected components (i.e., ``communities'' that make up a larger population) and each graph expansion can be processed independently across parallel threads before being their results are combined.
On an M4 Pro MacBook using 12 threads, a sample run for $k = 24$ and $d = 3$ takes approximately $7, 11, 16$ seconds when the observed graph has $50, 100, 150$ nodes, respectively.

\section{Evaluation} \label{evaluation}
\subsection{Dataset} \label{evaluation:dataset}
We build our experiments on a publicly available real-world HIV sexual contact network dataset with anonymized entries provided by the Inter-university Consortium for Political and Social Research (ICPSR) \cite{ICPSR} \footnote{Our codebase is publicly available at \url{https://github.com/akanga-harvard/policy-embedded-graph-expansion}. The ICPSR dataset is publicly available after completing the user agreement in the ICPSR portal \cite{ICPSR}.}. This dataset contains 778 nodes comprising 188 connected components and has an 11.3\% HIV prevalence rate. This small size is important for our setting since we expect similar data constraints as we deploy (see Section \ref{intro:collaboration}). Figure \ref{fig:dataset} shows an overview of the HIV sexual contact network.

\begin{figure}
    \centering
    \includegraphics[width=\linewidth]{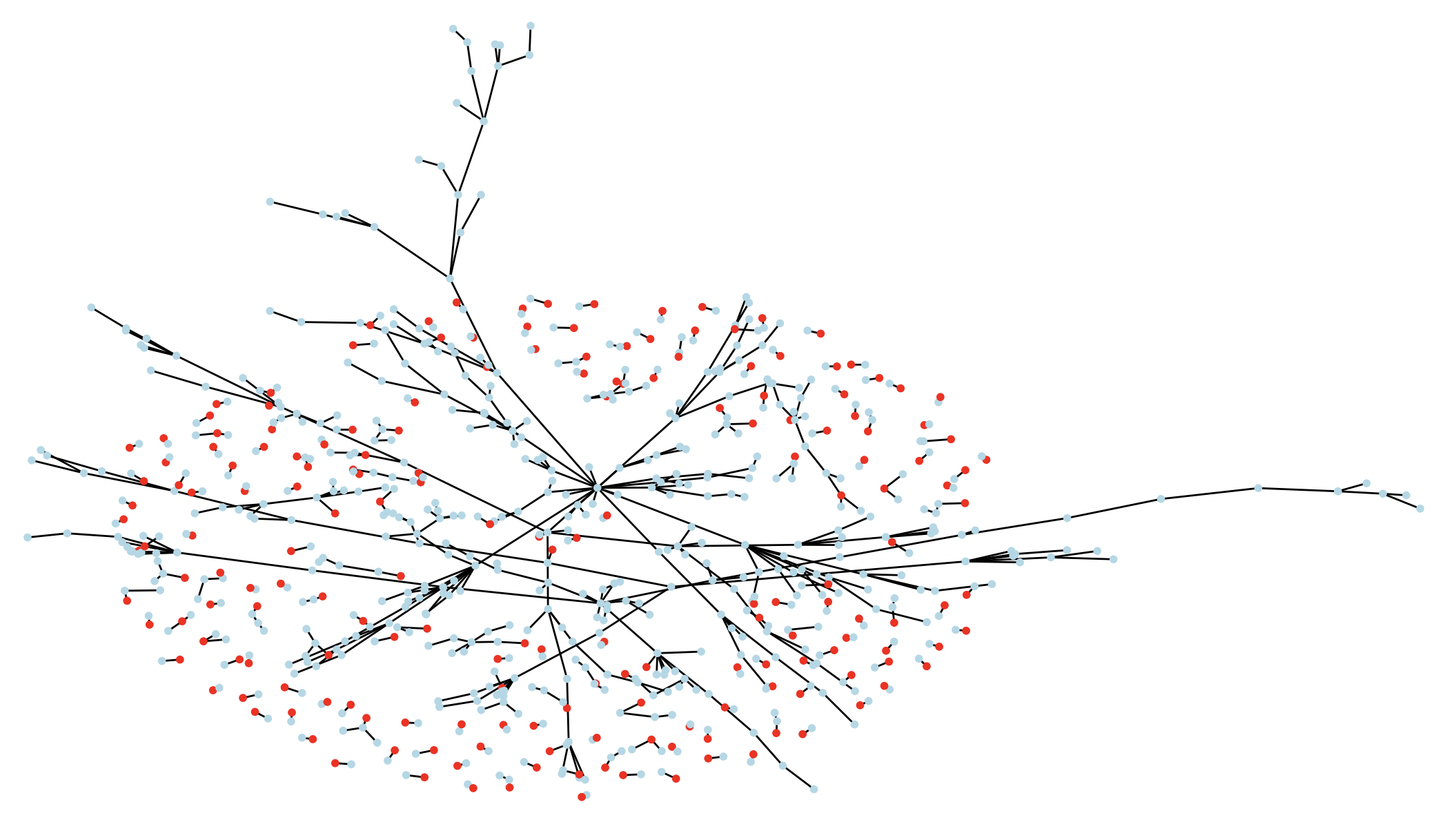}
    \caption{Overview of the ICPSR HIV sexual contact network dataset after forest projection (root nodes shown in red). Nodes in the graph represent people and edges represent sexual contacts.}
    \label{fig:dataset}
\end{figure}

Each node is associated with a 72-dimensional vector of binary covariates that are one-hot encoded from 17 categorical covariates related to gender, sexual orientation, occupation, etc. Additionally, to mirror the real-world data-collection process from our HIV prevention efforts in South Africa (see Section \ref{intro:collaboration}), we project each connected component in the test graph to a tree, forming a forest. In practice, edges are recorded as people refer their contacts to testing clinics and services. However, people are unlikely to (or may not be allowed to) accept a second referral to the same service, especially when referrals are limited. As such, the forest structure arises naturally and noise is introduced in the form of randomly unrecorded edges. These forest structures have been observed in other studies on peer-to-peer network recruitment as well (e.g., \cite{stein2018stochastic}). Furthermore, prior research has observed that actual STI transmission networks are often sparse and tree-like, and that the Gittins index-based policy remains robust to non-tree graphs \cite{choo2025adaptive}. 

We split the ICPSR HIV sexual contact network data into 5 partitions with an equal number of connected components such that no connected components are split up. We also ensure that partitions roughly match the distribution of connected component sizes that is present in the overall dataset. For each experiment in this section, we run 5 SAPOG instances with static test graphs and average the normalized results, weighting each instance by how many nodes are in the test graph. Each of the 5 data partitions serves as the test graph once and serves as part of the training data four times.

\subsection{Baseline Algorithms} \label{evaluation:baselines}
In addition to our method (PEGE + DDB), we evaluate five baseline algorithms for reference, which we describe below:
\hspace{0pt}\\
\textbf{1) Random (R):} Always select a random frontier node.
\\
\textbf{2) Greedy Neighbor (GN):} During each timestep, select a frontier node that has a parent node that tested positive. If no such node exists, return a random frontier node. This is an intuitive strategy that a human might apply.
\\
\textbf{3) Greedy Classifier (GC):} Use the training data to train a linear support vector classifier (SVC) with balanced class weights \cite{scikit-learn} to predict the probability that a given node tests positive conditioned on its covariates. A linear SVC was selected for the greedy baseline as it outperformed logistic regression, RBF-SVC, and gradient boosting. At test time, we always select the frontier node with the greatest probability of testing positive. This is a heuristic approach that is indifferent to any network-based data and is representative of an individual risk-based algorithm that has been proposed for real-world deployment (see Section \ref{intro:collaboration}).
\\
\textbf{4) Deep Q-Network (DQN):} We implement a DQN baseline \cite{mnih2015human} with the Pytorch NNConv neural architecture \cite{fey2019fast} that uses a message-passing GNN with edge-conditioned weights \cite{gilmer2017neural}. We allow the DQN to train on $\mathcal{G}$ and $\mathbf{X}$ (but not $\mathbf{Y}$), which serves as a conservative buffer against other learned policies.
\\
\textbf{5) Fully Observable Gittins (FOG)}: Apply the Gittins index-based policy to $\mathcal{G}, \mathbf{X}, \mathbf{Y}^{(t)}$ on each timestep and select the frontier node with the highest Gittins index. The feature mapping $\theta$ used in the Gittins index policy is allowed to train on the test set, meaning the oracle itself is also stronger in this case. Since this algorithm can observe the entire ground truth network, FOG serves as an unattainable upper bound for any policy with the same information described in our model.

\subsection{Experiments and Results}\label{evaluation:experiments}
We evaluate methods based on how much normalized discounted reward they accrue throughout the experiments (see \cref{eq:reward-model} for reward model), which captures our real-world goal of identifying as many HIV-positive individuals as possible as quickly as possible. First, we compare the average performance of our method (PEGE + DDB) to that of each baseline. For this experiment, we fix the expansion depth in DDB at $d = 3$, the number of graph expansions in PEGE at $k = 24$, the discount factor at $\gamma = 0.99$, and the number of tests at $T = |\mathbf{V}|$. We also cap the maximum expansion width at $b = 25$ to guarantee bounded graph expansions. The left side of Figure \ref{fig:core_experiment} shows the empirical results. Table \ref{tab:stats} also summarizes some performance statistics, including area under the cumulative discounted reward curve (AUC) and cumulative discounted reward at different testing budgets.



On average, our PEGE + DDB method outperforms all baselines. The total AUC for PEGE + DDB is about 9.42\% greater than that of the closest baseline (Greedy Classifier). Furthermore, PEGE + DDB remains advantageous under different testing budgets, outperforming the strongest baseline by 17.3\%, 8.06\%, and 6.83\% under testing budgets of 25\%, 50\%, and 75\% of the population, respectively. PEGE also captures about 78.4\% of the AUC of FOG, and 76.3\%, 76.6\%, and 80.6\% of the cumulative discounted reward accrued by FOG under testing budgets of 25\%, 50\%, and 75\% of the population, respectively. The remaining gap is likely due to the observability gap, the stronger oracle that FOG accesses, and error in DDB. These boosts that we observe in discounted reward also indicate more positive detections early on. For example, when testing only 25\% of the population, PEGE + DDB discovers about 15.4\% more positive cases than the strongest baseline. If network-based testing algorithms are eventually deployed with tens or hundreds of thousands of participants, improvements of this scale can translate to hundreds or thousands of additional positive HIV identifications, assuming the same population-wide positivity rate of 11.3\%.

Next, to isolate contributions of each component of our method, we present an ablation study comparing PEGE + DDB against Naive Gittins (no graph expansion), PEGE + DDB (Full Theta), which uses the $\theta$ trained on the full (train and test) network in the oracle, and the FOG upper bound. The results are shown in the right side of Figure \ref{fig:core_experiment}. We observe:
(i) \emph{Data Richness}: Training $\theta$ on the full network increases PEGE + DDB's AUC by 9.87\%, highlighting possible gains from larger datasets.
(ii) \emph{Framework Value}: PEGE + DDB outperforms the AUC of Naive Gittins by about 10.2\% (similar to improvements over other heuristics), isolating the framework's benefit from the oracle's baseline performance.
(iii) \emph{Early Impact}: DDB expansions are most useful in early rounds when graph uncertainty is highest. This makes the framework especially fit for resource-critical problems (e.g., at 10\% budget we outperform Naive Gittins by 66.7\%).

\begin{figure}[t]
    \centering
    \includegraphics[width=\linewidth]{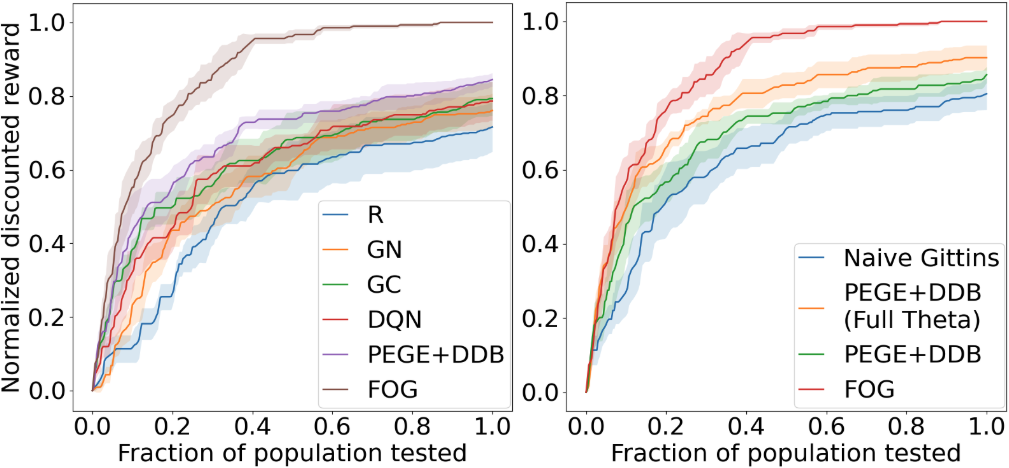}
    \caption{Left: Performance of our method (PEGE + DDB) compared to several baselines and the FOG upper bound, as measured by normalized cumulative discounted reward. Right: Performance of our PEGE + DDB method compared to Naive Gittins with no graph expansion, PEGE + DDB with $\theta$ trained on the full network in the oracle, and the FOG upper bound.}
    \label{fig:core_experiment}
\end{figure}

\begin{table}[t]
\centering
\begin{tabular}{@{}lcccc@{}}
\toprule
Algorithm  & AUC & 25\% & 50\% & 75\%\\
\midrule
Random            & 0.506  & 0.389  & 0.598 & 0.670\\
Greedy Neighbor   & 0.561  & 0.473  & 0.622 & 0.715\\
Greedy Classifier & 0.616  & 0.526  & 0.682 & 0.737\\
DQN               & 0.603  & 0.504  & 0.665 & 0.747\\
PEGE + DDB (Ours) & \textbf{0.674} & \textbf{0.617} & \textbf{0.737} & \textbf{0.798}\\
\emph{Fully Observable Gittins} & \emph{0.860} & \emph{0.809} & \emph{0.962} & \emph{0.990}\\
\bottomrule
\end{tabular}
\caption{Total AUC and normalized discounted reward at different testing budgets (25\%, 50\%, 75\% of population) for each algorithm. In each column, the best is \textbf{bolded} and the upper bound is \emph{italicized}.}
\label{tab:stats}
\end{table}

Next, we investigate some of the key parameters that drive the performance of our approach. First, we examine how varying $k$, the number of DDB graph expansions that we sample and aggregate in PEGE, affects performance. Second, we explore different methods of aggregating oracle scores to select actions in PEGE. While we used the mean in our PEGE specification in Section \ref{methods:pege}, the aggregation method is tunable. Specifically, we experiment with the following additional aggregation heuristics (with $k = 24$):
\hspace{0pt}\\
\textbf{1) Mode:} Define a $k$-length array $g$ containing the frontier node with the highest Gittins index from each of the $k$ synthetic graph expansions. Select the mode of $g$.
\\
\textbf{2) Mean + variance (explorative):} Select the frontier node with the greatest sum of mean and variance in Gittins index across the $k$ graph expansions.
\\
\textbf{3) Mean - variance (exploitative):} Select the frontier node with the greatest difference between mean and variance in Gittins index across the $k$ graph expansions.

Figure \ref{fig:ablation} shows the results of these two experiments, in which we again fix $d = 3$ and $\gamma = 0.99$. First, we observe that increasing $k$ from 1 to 12 results in an AUC gain of 7.01\%, highlighting the benefit of considering a distribution of sampled network completions. Next, we see that the performance of PEGE + DDB begins to plateau after $k = 12$, with fluctuations in AUC of less than 0.508\% for $k = 24$ and $k=36$, compared to $k=12$. This means that our approach is relatively sample efficient, which is important because the decision-making oracle $\mathcal{O}$ can be quite expensive. Even though the parallelizability of PEGE helps curb runtime, large numbers of samples are still consequential in terms of computing resources.

\begin{figure}
    \centering
    \includegraphics[width = \linewidth]{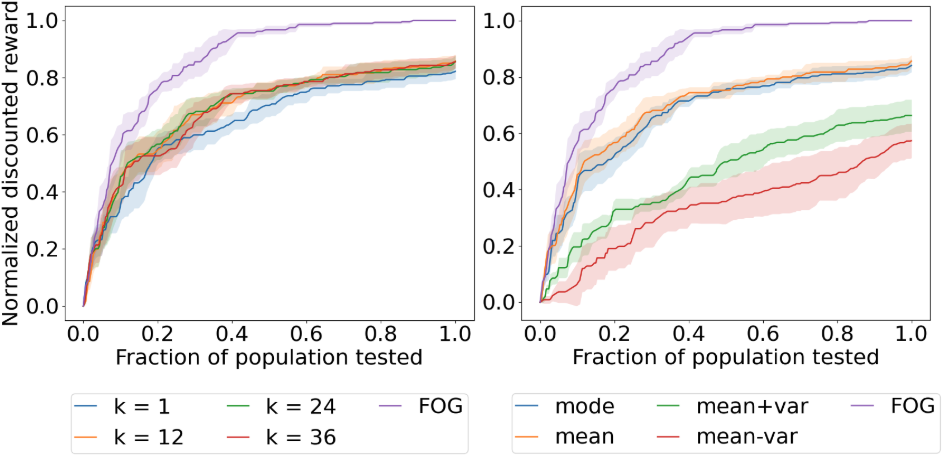}
    \caption{Cumulative discounted reward accrued by our PEGE + DDB method using different sample parameters $k$ (left) and Gittins index aggregation methods (right).}
    \label{fig:ablation}
\end{figure}

Furthermore, we observe that heuristic aggregation approaches fail to outperform the simple sample mean aggregation. In fact, the mean-var and mean+var aggregation methods record AUCs of only about 48.4\% and 65.4\% (respectively) of that of mean, demonstrating that striking a balance of exploration and exploitation is important in our setting. Additionally, mode-based heuristic aggregation also underperforms mean by about 2.81\% in terms of AUC.

\section{Conclusion} \label{conclusion}
In this work, we explore a sequential HIV testing problem over an incrementally revealed graph that addresses gaps in prior network-based HIV testing work. We propose a new approach (PEGE + DDB) and perform experiments on real-world HIV networks that show that it consistently outperforms all other baselines. In future work, we are studying related problems in HIV testing, such as referral and incentive distribution algorithms for peer-to-peer network recruitment \cite{pan2026adaptive} and batched HIV testing policies \cite{kong2026latent}. We are also working to explore how PEGE can be applied to other applications of SAPOG.

Our interdisciplinary team is motivated by our real-world efforts in South Africa, where we are working to deploy the algorithms that we demonstrate in this paper. The improvements that our approach provides can lead to hundreds or thousands of additional HIV detections when deployed at a larger scale of tens or hundreds of thousands of testing participants, which is crucial for getting people with HIV started with treatment, getting their high-risk contacts started with prevention like PrEP, and limiting further spread of the virus. This represents direct progress towards the first of the 95-95-95 targets for HIV \cite{unaids2022} and the UN Sustainable Development Goal 3.3 \cite{UN_SDG3}. 


\section*{Acknowledgements}
This work was supported by ONR MURI N00014-24-1-2742. The findings and conclusions in this report are those of the authors and do not necessarily represent the official position of the WHO.
\bibliographystyle{named}
\bibliography{ijcai26}

\end{document}